\documentclass[runningheads]{llncs}

\title{
  Pacing Types for Asynchronous Stream Equations
}

\usepackage{xspace}
\usepackage{listings}
\usepackage{xcolor}
\usepackage{lstautogobble}
\usepackage{color}
\usepackage{amsmath}
\usepackage{todonotes}
\usepackage{mathpartir}
\usepackage{subcaption}
\usepackage{float}
\usepackage{stmaryrd}
\usepackage{csquotes}
\usepackage{enumitem}
\usepackage{hyperref}
\usepackage{amssymb, amsfonts}
\usepackage{cleveref}
\usepackage{orcidlink}

\begin{document}
  \author{
    Florian Kohn\orcidlink{0000-0001-9672-2398} \and 
    Arthur Correnson\orcidlink{0000-0003-2307-2296} \and 
    Jan Baumeister\orcidlink{0000-0002-8891-7483} \and 
    Bernd Finkbeiner\orcidlink{0000-0002-4280-8441}}
  \authorrunning{F. Kohn et. al.}
  
  \institute{
    CISPA Helmholtz Center for Information Security, Saarbrücken, Germany \email{\{florian.kohn,arthur.correnson,jan.baumeister,finkbeiner\}@cispa.de}
  }
  \maketitle
  \begin{abstract}
    Stream-based monitoring is a runtime verification approach where
    a \emph{monitor} aggregates streams of input data from sensors and other sources to give real-time statistics and assessments of a system's health.
    One of the central challenges in designing reliable stream-based monitors is to deal with the asynchronous nature of data streams:
    in concrete applications, the different sensors being monitored produce values at different speeds, and it is the monitor's responsibility to correctly react to the asynchronous arrival of different streams of values.
    To ease this process, modern frameworks for stream-based monitoring such as \textsc{RTLola} enable users to finely specify data synchronization policies via a system of \emph{pacing annotations}.
    While this feature simplifies the design of monitors, it can also lead users to write \emph{inconsistent} policies, where synchronization between two streams is explicitly requested via annotations, but cannot always be achieved.
    To mitigate this issue, this paper presents \emph{pacing types}, a novel type system implemented in RTLola to ensure that monitors for asynchronous streams are free of timing inconsistencies.
    We give a formal semantics to pacing annotations for a core fragment of \textsc{RTLola}, and present a soundness proof of the pacing type system.
    For an additional level of guarantees, we machine-checked the soundness proof using the Rocq proof assistant.
  \end{abstract}

  \newcommand{\CLola}{$\mu$\textsc{RTLola}\xspace}
\newcommand{\stream}[1]{{\footnotesize\texttt{#1}}\xspace}
\newcommand{\bad}{\lightning}
\newcommand{\headline}[1]{\paragraph*{\textbf{#1.}}}

\definecolor{bluekeywords}{rgb}{0.13, 0.13, 1}
\definecolor{greentypes}{rgb}{0, 0.5, 0}
\definecolor{orangecomments}{rgb}{1, 0.5, 0.1}
\definecolor{redstrings}{RGB}{171, 114, 2}
\definecolor{graynumbers}{rgb}{0.5, 0.5, 0.5}
\definecolor{goldcomments}{rgb}{0.6, 0.4, 0.08}

\lstdefinelanguage{Lola}{
  keywords=[0]{input, output, trigger, constant, import, spawn, eval, close, with, when},
  moredelim=**[is][\color{greentypes}@]{@}{@},
  keywordstyle=[0]\bfseries\color{bluekeywords},
  keywords=[1]{if, then, else, aggregate, defaults, offset, last, prev, by, or, to, sin, cos, abs, hold, over, using, over_instances},
  keywords=[2]{Variable, String, Int, Int64, UInt, UInt64, Bool, Float32, Float64, Float},
  keywordstyle=[2]\color{greentypes},
  sensitive=false,
  comment=[l]{//},
  morecomment=[s]{/*}{*/},
  morestring=[b]',
  morestring=[b]",
  literate={\\@}{@}1
}
\lstset{
    autogobble,
    language={Lola},
    columns=fullflexible,
    showspaces=false,
    showtabs=false,
    breaklines=true,
    showstringspaces=false,
    breakatwhitespace=true,
    escapeinside={(*@}{@*)},
    commentstyle=\color{orangecomments},
    keywordstyle=[1]\color{bluekeywords},
    stringstyle=\color{redstrings},
    numberstyle=\color{graynumbers},
    basicstyle=\ttfamily\small,
    frame=l,
    framesep=5pt,
    xleftmargin=15pt,
    tabsize=4,
    captionpos=b,
    mathescape,
    numbers=left,
    stepnumber=1,
}

  \section{Introduction}

Runtime verification is a safety assurance mechanism in which systems are monitored during their execution. 
Runtime verification is particularly useful for increasing the reliability of complex cyber-physical systems (CPS) for which obtaining a model for static verification is difficult or impossible.
Stream-based monitoring is an approach to runtime verification in which a monitor
collects data through input streams, which are then combined and aggregated in output streams to give real-time assessments of a system's health.


While, there exist many formalisms and programming languages to describe stream-based monitors \cite{d2005lola,DBLP:conf/cav/FaymonvilleFSSS19,volostream,DBLP:conf/rv/KallwiesLSSTW22,DBLP:conf/fm/GorostiagaS21,DBLP:conf/fm/PerezGD24}, we focus on the stream-based specification language \textsc{RTLola}.
In \textsc{RTLola}, a monitor is specified by listing equations relating input and output streams.
Certain boolean output streams, called triggers, indicate whether the system being monitored is behaving as expected.

From a system perspective, monitors are passive components.
They react to inputs from the system as they arrive.
If different system sensors generate inputs at varying rates, the monitor must handle the asynchronous arrival of data.
To simplify this process, \textsc{RTLola} features a system of \emph{pacing annotations} where
each output stream is labeled with additional information indicating when a stream must
produce a fresh value in reaction to the asynchronous arrival of new input values.
While the pacing annotation system of \textsc{RTLola} has proven to be
useful in practical applications \cite{Lola2.0,baumeister2020rtlola,volostream}, it also introduces a critical challenge: ensuring the \emph{consistency} of pacing annotations.
Indeed, it is easy to write annotations that are incompatible with one another
or that over-constrain a stream by requiring it to produce values at times when its dependencies are not guaranteed to be available.
Ensuring consistency of pacing annotations is all the more difficult because the semantics of \textsc{RTLola}'s pacing annotations have, thus far, never been rigorously formalized.
In this paper, we address this issue in two ways.
We propose the first formal semantics of \textsc{RTLola}'s pacing annotations, and
we present a type system to check the consistency of pacing annotations.

\paragraph{Contributions}
In summary, we make the following contributions:
\begin{enumerate}
    \item We formalize the essence of \textsc{RTLola}'s system of pacing annotations.
    In particular, we give a formal semantics to stream equations decorated with pacing annotations.
    \item We present a \emph{type system} to check the consistency of pacing annotations, and formally prove its soundness. Concretely, we prove that well-typed annotated equations have at least one solution.
    \item Finally, to obtain an additional level of reliability, we machine-checked the correctness proof of the type system using the Rocq proof assistant.
\end{enumerate}

\paragraph{Structure of the paper}
\Cref{sec:overview} provides an overview of \textsc{RTLola} and its system of pacing annotations. It also illustrates the inconsistency problem.
\Cref{sec:pacing_types} formalizes semantics of pacing annotations for a core fragment of \textsc{RTLola}, and introduces a type system to check consistency of pacing annotations.
\Cref{sec:proof} presents a soundness proof of the type system.
Finally, \Cref{sec:limitations} describes two extensions of the type system.

\subsection{Related Work}
Runtime verification has been studied since the late 90s~\cite{DBLP:conf/ecrts/KimVBKLS99}.
Specifications are often expressed in temporal logics such as LTL~\cite{bauer2011runtime} and MTL~\cite{DBLP:conf/cav/BasinKZ17}.
Applications include Markov Decision Processes~\cite{DBLP:conf/cav/JungesTS20,DBLP:conf/cav/HenzingerKKM23}, cyber-physical systems~\cite{volostream}, and software~\cite{DBLP:conf/icse/JinMLR12}.
Stream-based monitoring framework like Lola~\cite{d2005lola}, \textsc{RTLola}~\cite{DBLP:conf/fm/BaumeisterFKS24}, Tessla~\cite{DBLP:conf/rv/KallwiesLSSTW22}, and HStriver~\cite{DBLP:conf/fm/GorostiagaS21} provide expressive specification languages coupled
with efficient monitoring infrastructures.

Synchronous programming languages such as LUSTRE~\cite{DBLP:journals/pieee/HalbwachsCRP91}, ESTEREL~\cite{DBLP:conf/concur/BerryC84}, and Z\'elus~\cite{DBLP:conf/emsoft/BenvenisteBCP11} share similarities with stream-based languages for monitoring.
In particular, both paradigms need to address timing inconsistencies (ensuring that values exists when they are requested).
Synchronous languages use notions of \emph{clocks}~\cite{DBLP:conf/fpca/GautierG87,DBLP:conf/emsoft/ColacoP03,DBLP:conf/mpc/MandelPP10}, and associated static analyses to resolve timing issues. We propose a \emph{type analysis} to address the timing inconsistencies introduced by the expressive pacing annotations of \textsc{RTlola}.

We note that type-based approach have already been explored in similar contexts.
For example,~\cite{DBLP:conf/hybrid/BenvenisteBCPP14} introduces type systems to detect \emph{non-causal} equations, a property related to the existence of models in stream-based monitoring.
Other type systems address additional program properties; for example, MARVeLus employs a refinement type system that reasons about temporal behavior using LTL~\cite{DBLP:conf/focs/Pnueli77} specifications~\cite{DBLP:journals/pacmpl/ChenMABJSWZJ24}.

  \section{Asynchronous Streams in RTLola}\label{sec:overview}

This section provides a high-level overview of \textsc{RTLola} and its pacing annotations using a battery management system as an example.
In \textsc{RTLola}, monitors are specified by declaring a list of named input and output streams.
As an example, consider the following \textsc{RTLola} specification which 
monitors excessive power draw:
\begin{center}
  \begin{lstlisting}[label=ex:starting]
input  battery_lvl
output drain   := battery_lvl.prev(or: battery_lvl) - battery_lvl
output warning := drain > 5
  \end{lstlisting}
\end{center}

The input stream \stream{battery\_lvl} captures the current battery level measurement.
The output stream \stream{drain} computes the charge loss by subtracting the current battery level from the previous one, accessed via \lstinline|battery_lvl.prev(or:...)|.
The argument after \lstinline|or:| provides a default value for the first evaluation step.
This example illustrates the core concepts of input and output streams and the use of stream offsets to express temporal properties.

\headline{Handling Asynchronous Inputs}

In practice, a monitor often requires access to multiple sensors to produce meaningful verdicts about the system's health.
For example, the previous specification can be refined to monitor that the battery does not overheat while it is charging.
For this, an additional input stream \stream{temperature} captures the current temperature of the battery being monitored:
\begin{center}
  \begin{lstlisting}[label=ex:async_spec]
    input battery_lvl: Int
    input temperature: Int
    output drain := battery_lvl.prev(or: battery_lvl) - battery_lvl
    output warning := drain < 0 && temperature > 50
  \end{lstlisting}
\end{center}

Importantly, \textsc{RTLola} does not make the assumption that streams are synchronized: different inputs might arrive at different rates.
This raises the question of how to interpret an expression like \lstinline{drain < 0 && temperature > 50} in the previous specification.
More precisely, it is not clear what the value of this expression should be if one of the two inputs is not available at the time of evaluation.
For example, the monitor could take the last defined value of the stream or wait for all values to be available.
However, both solutions would negatively impact the precision of the monitor.
Implicitly taking the last defined values would allow to mix
measurements performed at very different time points and whose comparison is, therefore, difficult to interpret.
Similarly, waiting for all values to arrive might result in dangerous situations being missed.
Instead, \textsc{RTLola} permits more fine-grained handling of the asynchronous arrival of data via a combination of \emph{pacing annotations} and a selection of synchronous and asynchronous access operators.


\headline{Synchronous and Asynchronous Accesses}


To facilitate the handling of multiple inputs arriving at different paces, \textsc{RTLola} distinguishes between \emph{synchronous} and \emph{asynchronous} stream accesses.
A synchronous access requires the stream being accessed (i.e., the stream the access refers to) to have a value at the same time as the accessing stream (i.e., the stream that contains the access in its expression).
Direct accesses (such as \lstinline{drain} and \lstinline{temperature} in the previous example) and \lstinline{prev} accesses are synchronous by design.
This means that in the previous example, the stream \lstinline{warning}
can only produce a value when \emph{both} \lstinline{drain} and \lstinline{temperature} are also producing a new value simultaneously.
This restriction can be problematic if different sensors generate values at arbitrary, non-overlapping time points.
In particular, this could prevent \stream{warning} from ever being evaluated, thus never issuing the intended warning.


To address this issue, \textsc{RTLola} also provides an asynchronous \lstinline|hold| access.
A \lstinline|hold| access retrieves the last produced value of the accessed stream, regardless of when it was generated.
If no such value exists, a default value is returned, similar to \lstinline|prev|.
Unlike \lstinline|prev| though, the \lstinline|hold| operator may also reference the current value if the accessed stream produces a value at the same time as the accessing stream.
Using a \lstinline|hold| access the specification in \Cref{ex:async_spec} can be refined as follows: \begin{center}
  \begin{lstlisting}[label=ex:async_spec_hold]
    input battery_lvl: Int
    input temperature: Int
    output drain := battery_lvl.prev(or: battery_lvl) - battery_lvl
    output warning := drain < 0 && temperature.hold(or: 0) > 50
  \end{lstlisting}
\end{center}

Replacing the synchronous access to \stream{temperature} in \stream{warning} with a \lstinline|hold| eliminates the requirement for \stream{temperature} and \stream{battery\_lvl} to update simultaneously, allowing \stream{warning} to be evaluated whenever a new battery level is measured.
Figures \ref{fig:timing_diagram:sync} and \ref{fig:timing_diagram:hold} illustrate the difference between synchronous and asynchronous access on an example trace.
Black continuous arrows represent synchronous dependencies, while light gray arrows indicate dependencies due to default values and dashed arrows indicate asynchronous dependencies.
On the left, the stream
\stream{warning} accesses \stream{temperature} \emph{synchronously}, via \lstinline|prev|.
On the right, \stream{temperature} is accessed \emph{asynchronously} via \lstinline|hold|.

\begin{figure}[h]
  \begin{subfigure}{0.48\textwidth}
    \includegraphics[width=\linewidth]{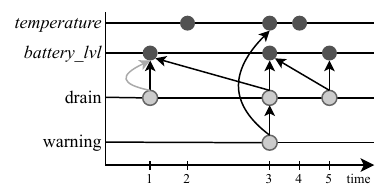}
    \caption{Synchronous access to \lstinline|temperature|}
    \label{fig:timing_diagram:sync}
  \end{subfigure}
  \hfill
  \begin{subfigure}{0.48\textwidth}
      \includegraphics[width=\textwidth]{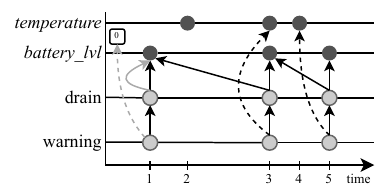}
      \caption{Asynchronous access to \lstinline|temperature|}
      \label{fig:timing_diagram:hold}
  \end{subfigure}
  \caption{Synchronous VS asynchronous accesses.}
  \label{fig:timing_diagram}
\end{figure}

We note that with a synchronous access, \stream{warning} only produces a value at time three, when both \stream{temperature} and \stream{battery\_lvl} receive new values.
Using an asynchronous access instead, \stream{warning} also produces values at steps one and five in addition to time step three, resulting in a more precise verdict.


\headline{Pacing Annotations}
One could argue that, the intent of the \stream{warning} stream is still not fully captured.
Indeed, even with the help of asynchronous \lstinline|hold| access, \stream{warning} only updates when the battery level changes, missing potential safety issues at times two and four in \Cref{fig:timing_diagram:hold}.
As a complement to asynchronous accesses, \textsc{RTLola} features \emph{pacing annotations}.
The idea is that each output stream will be labeled with an annotation 
indicating when the stream is expected to evaluate.
Concretely, these annotations consists of (positive) boolean formulas over input streams.
Our running example can be annotated as follows:
\begin{lstlisting}
  input battery_lvl: Int
  input temperature: Int
  output drain @battery_lvl@ := 
      battery_lvl.prev(or: battery_lvl) - battery_lvl
  output warning @battery_lvl | temperature@ := 
      drain.hold(or: 0) < 0 && temperature.hold(or: 0) > 50
\end{lstlisting}

Here, the output stream \stream{drain} is annotated with \lstinline|@battery_lvl|, indicating that it is expected to evaluate each time a new \stream{battery\_lvl} input is received.
The case of the \stream{warning} stream is more interesting:
we annotated it with \lstinline!@battery_lvl | temperature!, indicating that
it should be updated whenever either the temperature \emph{or} the battery level changes.
We further updated the specification to exclusively employ \lstinline|hold| accesses.
As a result of these modifications, the updated specification now produces a \emph{meaningful} temperature warning at \emph{all} time points (see \Cref{fig:timing_diagram:hold_hold}).

\begin{figure}[H]
  \centering
  \includegraphics[width=.5\textwidth]{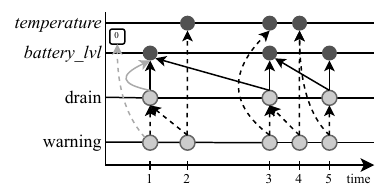}
    \caption{The effect of disjunctive pacing types}
    \label{fig:timing_diagram:hold_hold}
\end{figure}

We note that hold accesses like \lstinline|temperature.hold(or: 0)| are guaranteed to produce a value, even if the accessed stream \stream{temperature} has not updated in a very long time. As a result, over-using \lstinline|hold| can degrade the quality of verdicts by producing values that are not connected anymore to the current state of the monitored system.
In the previous example, the pacing annotation \lstinline!battery_lvl | temperature! ensures that the two \lstinline{hold} accesses produce values that are reasonably up-to-date with the current status of the system: a new warning value is only produced when one of the two monitored sensors updates.
Combining pacing annotations with asynchronous accesses in this manner enables balancing the frequency of updates and the interpretability of the calculated value

\subsection{The Problem}\label{sec:problem}

While pacing annotations offer a convenient way for users to finely
specify data-synchronization policies, they also enable users to write
\emph{inconsistent} specifications.
As an example, consider the specification in \Cref{fig:invalid:spec}.
\begin{figure}[H]
  \centering
  \begin{subfigure}{0.4\textwidth}
    \centering
    \begin{lstlisting}
      input a: Int
      input b: Int
      output x @b@ := b
      output y @a@ := x
    \end{lstlisting}
    \caption{An inconsistent specification}
    \label{fig:invalid:spec}
  \end{subfigure}
  \qquad
  \begin{subfigure}{0.4\textwidth}
      \centering
      \includegraphics[width=.7\textwidth]{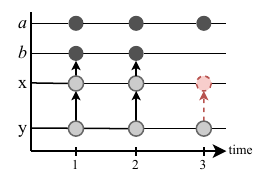}
      \caption{The timing problem visualized}
      \label{fig:invalid:timing}
  \end{subfigure}
  \caption{An example of inconsistent specification}
  \label{fig:invalid}
\end{figure}
The pacing annotation \lstinline|x @b| requires that \stream{x} evaluates whenever \stream{b} receives a new value. Similarly, the annotation \lstinline|y @a|
requires \stream{y} to update whenever \stream{a} receives a new value.
As demonstrated in \Cref{fig:invalid:timing}, these requirements cannot
always be met at runtime!
At time point three, the stream \stream{y} is expected to have a value (because \stream{a} also has a value).
Since \stream{y} contains a direct access to \stream{x}, it cannot evaluate (because \stream{x} does \emph{not} have a value at time point three).
The goal of the type system presented in this paper is to reject such invalid specifications, ensuring that a stream value \emph{can} always be computed when it \emph{must} be computed.

We note that one way to \enquote{repair} the specification in \Cref{fig:invalid:spec} would be to replace the synchronous access to \stream{x} in \stream{y} with a \lstinline|hold| access.
However, this change alters the semantics by removing the requirement that \stream{x} must produce a new value whenever \stream{y} produces a new value.
With the \lstinline|hold| operator, \stream{y} keeps producing fresh values (based on old values of \stream{x}) even if \stream{x} stops updating.

\subsection{A Type-driven Solution}

To prevent timing inconsistencies, \textsc{RTLola} uses a type system.
First, we observe that the asynchronous \lstinline|hold| stream access never causes timing inconsistencies, as it inherently provides a default value when the accessed value is absent.
Consequently, the type system focuses on synchronous stream accesses, such as direct or \lstinline|prev| accesses.
More precisely, the type system ensures that a synchronously accessed value is always available when needed at runtime.
To enforce this property, it requires that a stream accessed synchronously evaluates at least at the same time points as the stream containing the access.




  \section{Pacing Types}
\label{sec:pacing_types}

\newcommand{\Vin}{\mathbb{V}_\mathit{in}}
\newcommand{\Vout}{\mathbb{V}_\mathit{out}}
\newcommand{\syn}[1]{{\color{brown}\textbf{\texttt{#1}}}}
\newcommand{\stradd}[2]{#1 \ \syn{+} \ #2}
\newcommand{\streq}[2]{#1 \ \syn{==} \ #2}
\newcommand{\strand}[2]{#1 \ \syn{\&\&} \ #2}
\newcommand{\strnot}[1]{\syn{not} \ #1}
\newcommand{\strlt}[2]{#1 \ \syn{<} \ #2}
\newcommand{\strlast}[2]{#1 \syn{.\texttt{prev}(}#2 \syn{)}}
\newcommand{\strhold}[2]{#1 \syn{.\texttt{hold}(}#2 \syn{)}}
\newcommand{\strdef}[3]{#1 \ \syn{@} \ #2 \ \syn{:=} \ #3}
\newcommand{\strwhen}[4]{#1 \ \syn{@} \ #2 \ \syn{when} \ #3 \ \syn{:=} \ #4}
\newcommand{\rhoIn}{{\rho_\textit{in}}}
\newcommand{\rhoOut}{{\rho_\textit{out}}}
\newcommand{\last}[2]{\mathit{Last}(#1, #2)}
\newcommand{\sync}[2]{\mathit{Sync}(#1, #2)}
\newcommand{\prev}[3]{\mathit{Prev}(#1, #2, #3)}
\newcommand{\hold}[3]{\mathit{Hold}(#1, #2, #3)}
\newcommand{\den}[1]{\llbracket #1 \rrbracket}

To formalize pacing annotations and their properties,
we focus on a core language modeling the key features of \textsc{RTLola}, here dubbed \CLola (for micro \textsc{RTLola}).
In essence, \CLola programs represent a set of
constraints over input and output \emph{stream variables}
drawn from (disjoint) sets $\Vin$ and $\Vout$ respectively.
For a fixed set of (named) input streams, these
constraints describe which set of (named) output
streams are valid reactions to the inputs.
Because of this declarative view, we refer to \CLola programs as \emph{specifications}.
For clarity, we only formalize constraints over integer-valued streams, but \CLola can be extended to arbitrary data types in a straightforward manner\footnote{In fact, in our Rocq formalization, \CLola is parameterized by an arbitrary type of values equipped with arbitrary (total) operators.}.

\subsection{Syntax of \CLola}

Constraints on stream variables are expressed via a list of equations of the form
$\strdef{x}{\tau}{e}$. 
In these equations, $x \in \Vout$ is the name of an output stream being constrained, and $e$ is a \emph{stream expression} over input and output variables. Finally, $\tau$ is a positive boolean formula over input variables describing when the output stream $x$ must be evaluated depending on the availability of inputs.
We call $\tau$ a \emph{pacing annotation}.

Stream expressions can be constant values $v \in \mathbb{Z}$, stream variables $x \in \Vin \uplus \Vout$, arithmetic functions applied to sub-expressions, references to previous values $\strlast{x}{e}$, or \emph{hold} accesses $\strhold{x}{e}$.
In $\strlast{x}{e}$ and $\strhold{x}{e}$,
$x \in \Vin \uplus \Vout$ is the name of an arbitrary input or output stream, and $e$ is a \emph{default expression} representing what value the expression should take if there is no previous value for $x$ (in the case of \syn{prev}) or if there is no value to hold (in the case of \syn{hold}). The exact syntax of specifications, stream expressions, and pacing annotations is defined below.

\begin{center}
  \begin{tabular}{ll}
    \textit{Variables}          & \quad $x \in \mathbb{V}_\texttt{in} \uplus \mathbb{V}_\texttt{out}$\\
    \textit{Values}             & \quad $v \in \mathbb{Z}$\\
    \textit{Stream Expressions} & \quad $e \ ::= \ v \mid x \mid \strlast{x}{e} \mid \strhold{x}{e} \mid \stradd{e_1}{e_2} \mid \ldots$\\
    \textit{Pacing Annotations} & \quad $\tau \ ::= \ x \mid \syn{$\top$} \mid \tau_1 \ \syn{$\wedge$} \ \tau_2 \mid \tau_1 \ \syn{$\vee$} \ \tau_2 $\\
    \textit{Equations}          & \quad $eq ::= \ \strdef{x}{\tau}{e}$\\
    \textit{Specifications}     & \quad $S \ ::= \ \epsilon \mid eq \cdot S$
  \end{tabular}
\end{center}

\subsection{Semantics of \CLola}
\label{sec:semantics}

\headline{Streams}

We equip \CLola with a denotational semantics assigning to lists of equations $S$ a relation $\den{S}$ between input and output streams. Importantly, to account for the asynchronous arrival of inputs (and the asynchronous production of outputs), streams are defined as sequences of \emph{optional} values: at each point in time, a stream can either have a value $v \in \mathbb{Z}$, or be undefined denoted by $\bot$. Formally, the set of streams is defined as follows: \[
  \mathit{Stream} \triangleq (\mathbb{Z} \uplus \{ \bot \})^\mathbb{N}
\]

When working with streams of optional values, it is often useful to equate streams that are not exactly equal, but only equal at non-$\bot$ positions. Given two streams $w_1, w_2 \in \mathit{Stream}$, we define the relation $w_1 \preccurlyeq w_2$ (read ``$w_1$ \textit{is less defined than} $w_2$'') as
\[
  w_1 \preccurlyeq w_2 \iff \forall n. \ w_1(n) = \bot \vee w_1(n) = w_2(n)
\]

\noindent Intuitively, $w_1 \preccurlyeq w_2$ 
guarantees that $w_1$ and $w_2$ are equal at positions where $w_1$ is defined, but $w_2$ is unconstrained at positions where $w_1$ contains $\bot$.

\headline{Maps}

\CLola equations are interpreted with respect to \emph{stream maps}
associating each input and output variable to a stream.
We note $\mathit{InMap}$ and $\mathit{OutMap}$ for input and output stream maps. \[
  \mathit{InMap} \triangleq \mathit{Stream}^{\Vin} \qquad \mathit{OutMap} \triangleq \mathit{Stream}^{\Vout}
\]

When proving the soundness of the type system, it will also be useful to consider \emph{partial} output stream maps. We represent them as total maps to optional streams. \[
  \mathit{ParMap} \triangleq (\mathit{Stream} \cup \{ \ \bot \ \})^{\Vout}
\]

Given a partial map $\rho \in \mathit{ParMap}$, we note $\mathit{dom}(\rho) \subseteq \Vout$ for the domain of $\rho$ (i.e. variables $x$ such that $\rho(x) \ne \bot$), $\emptyset$ for the empty map sending all variables to $\bot$, and $(x \mapsto w)$ the singleton map sending $x$ to the stream $w$.
Finally, given a partial map $\rho_\mathit{par} \in \mathit{ParMap}$ and a total map 
$\rho_\mathit{tot} \in \mathit{OutMap}$, we define the \emph{totalization}
of $\rho_\mathit{par}$ with respect to $\rho_\mathit{tot}$ as follows \[
  \rho_\mathit{par} \cdot \rho_\mathit{tot} \triangleq \lambda x. \ \begin{cases}\begin{tabular}{ll}
    $\rho_\mathit{par}(x)$ &if $x \in \mathit{dom}(\rho_\mathit{par})$\\
    $\rho_\mathit{tot}(x)$ &otherwise
  \end{tabular}\end{cases}
\]

In the rest of the paper, $\rhoIn$ and $\rhoOut$ are variables ranging over $\mathit{InMap}$ and $\mathit{OutMap}$ respectively. We also use
the notation $\rhoOut$ for variables ranging over partial output maps.

\headline{Semantics of Stream Expressions}

Given a pair of input and output stream maps $(\rhoIn, \rhoOut)$, an expression $e$ denotes a stream $\den{e}_{(\rhoIn, \rhoOut)} \in \mathit{Stream}$.
To assign a meaning to stream operators, we first introduce the following functions $\mathit{Last}$, $\mathit{Prev}$, and $\mathit{Hold}$.
In all definitions below, $w \in \mathit{Stream}$ is a stream of optional values, $n \in \mathbb{N}$ is a time point and $v \in \mathbb{Z} \cup \{ \bot \}$ is an optional value.

\begin{definition}[Stream operators]
  \upshape
  \begin{align*}
    \last{w}{n} &\triangleq \begin{cases}
        \begin{tabular}{ll}
          $w(n)$ &if $\ w(n) \neq \bot$\\
          $\last{w}{n-1}$ &if $\ w(n) = \bot \land n > 0$\\
          $\bot$ &otherwise
        \end{tabular}
      \end{cases}\\
      \hold{w}{n}{v} &\triangleq \begin{cases}
        \begin{tabular}{ll}
          $\last{w}{n}$ &\quad \ \ if $\last{w}{n} \neq \bot$\\
          $v$ &\quad \ \ otherwise
        \end{tabular}
      \end{cases}\\
      \prev{w}{n}{v} &\triangleq \begin{cases}
        \begin{tabular}{ll}
          $\bot$ &if $w(n) = \bot$\\
          $\last{w}{n-1}$ &if $w(n) \neq \bot \land n > 0 \land \last{w}{n-1} \neq \bot$\\
          $v$ &otherwise\\
        \end{tabular}
      \end{cases}
  \end{align*}
\end{definition}

\noindent The intuition behind these functions can be described as follows: \begin{itemize}[leftmargin=*]
  \setlength\itemsep{.4em}
  \item $\last{w}{n}$ is the last defined value of $w$ before the time point $n$ (included), or $\bot$ if $w$ is never defined before $n$.
  \item $\hold{w}{n}{v}$ models \syn{hold} accesses.
  It returns the last value of $w$ before $n$ (included), or a default value $v$
  if $w$ is never defined before $n$.
  \item $\prev{w}{n}{v}$ models \syn{prev} accesses.
  Since \syn{prev} accesses are synchronous, $\mathit{Prev}$ returns $\bot$ if
  $w(n)$ is undefined.
  If $w(n)$ is defined, it returns the last defined value before $n$ (excluded), or a default value $v$ if $w$ is never defined before $n$.
\end{itemize}

\noindent Given a pair of stream maps $\rho = (\rhoIn, \rhoOut)$ and a stream variable $x \in \Vin \uplus \Vout$, we use the notation convention $\rho(x)$ to denote the input or output stream associated with $x$:
\[
  \rho(x) \triangleq \begin{cases}
    \begin{tabular}{ll}
      $\rhoIn(x)$ &if $x \in \Vin$\\
      $\rhoOut(x)$ &if $x \in \Vout$
    \end{tabular}
  \end{cases}
\]
Using these definitions, we can introduce the denotation of stream expressions.

\begin{definition}[Denotation of stream expressions]
  \upshape
  \begin{align*}
    \den{c}_\rho(n) &\triangleq c\\
    \den{\stradd{e_1}{e_2}}_\rho(n) &\triangleq \den{e_1}_\rho(n) +_{\bot} \den{e_2}_\rho(n)\\
    \den{x}_\rho(n) &\triangleq \rho(x)(n)\\
    \den{\strlast{x}{e}}_\rho(n) &\triangleq \prev{\rho(x)}{n}{\den{e}_\rho(n)}\\
    \den{\strhold{x}{e}}_\rho(n) &\triangleq \hold{\rho(x)}{n}{\den{e}_\rho(n)}
  \end{align*}
  Where the operator $+_{\bot}$ is defined as follows:\[
  v_1 +_{\bot} v_2 \triangleq \begin{cases}
    \begin{tabular}{ll}
      $v_1 + v_2$ &if $v_1 \neq \bot \land v_2 \neq \bot$\\
      $\bot$ &otherwise
    \end{tabular}
  \end{cases}
\]
\end{definition}

\headline{Semantics of Pacing Annotations}

Given an input map $\rhoIn$, a pacing annotation $\tau$ is
interpreted as a set of time points $\den{\tau}_\rhoIn \subseteq \mathbb{N}$.
Concretely $\den{\tau}_\rhoIn$ represents all of the time points where the inputs described by the boolean formula $\tau$ are defined in $\rhoIn$.

\begin{definition}[Denotation of pacing annotations]
  \upshape
  \begin{align*}
    \den{x}_\rhoIn &\triangleq \{ \ n \mid \rhoIn(x)(n) \ne \bot \ \}\\
    \den{\syn{$\top$}}^n_\rhoIn &\triangleq \mathbb{N}\\
    \den{\tau_1 \ \syn{$\vee$} \ \tau_2}_\rhoIn &\triangleq \den{\tau_1}_\rhoIn \cup \den{\tau_2}_\rhoIn\\
    \den{\tau_1 \ \syn{$\wedge$} \ \tau_2}_\rhoIn &\triangleq \den{\tau_1}_\rhoIn \cap \den{\tau_2}_\rhoIn
  \end{align*}
\end{definition}

Pacing annotation $\tau$ can also be understood as describing the set of streams which have defined values exactly at time points $\den{\tau}_\rhoIn$.
We write $\den{\tau}^\omega_\rhoIn \subseteq \mathit{Stream}$ for the set of such ``$\tau$-paced'' streams.

\begin{definition}[Well-paced streams]
  \begin{align*}
    \den{\tau}^\omega_\rhoIn \triangleq \{ \ w \mid \forall n. \ n \in \den{\tau}_\rhoIn \iff w(n) \ne \bot \ \}
  \end{align*}
\end{definition}

\headline{Semantics of Specifications}

Having defined the semantics of expressions and pacing annotations, we can define the semantics of equations and specifications. An equation $\strdef{x}{\tau}{e}$ denotes a binary relation between input and output maps.
A specification (i.e., a list of equations) is interpreted as the intersection of the solutions of its equations.

\begin{definition}[Denotation of equations]
  \upshape
  \begin{align*}
    \den{\strdef{x}{\tau}{e}} &\triangleq \{ \ (\rhoIn, \rhoOut) \mid \rhoOut(x) \in \den{\tau}^\omega_\rhoIn \wedge \rhoOut(x) \preccurlyeq \den{e}_{(\rhoIn, \rhoOut)} \ \}\\
    \den{ \mathit{eq}_1 \cdot \ldots \cdot \mathit{eq}_n } &\triangleq \den{\mathit{eq}_1} \cap \ldots \cap \den{\mathit{eq}_n}
  \end{align*}
\end{definition}

A pair of stream maps $(\rhoIn, \rhoOut)$ satisfies a single equation $\strdef{x}{\tau}{e}$ 
under two conditions. First, the stream $x$ 
should be well-paced w.r.t. $\tau$ (i.e.,$\rhoOut(x) \in \den{\tau}^\omega_\rhoIn$).
Second, $x$ should be equal to its definition $e$.
We note that the pacing annotation can force $x$ to evaluate to $\bot$ at certain time point, even though $e$ would have had a defined value. Therefore, we only ask $x$ to be less defined than $e$ rather than strictly equal (i.e., $\rhoOut(x) \preccurlyeq \den{e}_{(\rhoIn, \rhoOut)}$).

\headline{Consistency}

Equations might have no solutions, or only admit solutions for specific combination of inputs.
As discussed in the overview, this issue is worsened by the addition of pacing annotations.
Indeed, in the semantics of equations, the additional pacing constraint ($\rhoOut(x) \in \den{\tau}^\omega_\rhoIn$) can turn solvable equations into unsolvable ones.
For example, assuming $a$ and $b$ are two input variables, the equations \begin{align*}
  &\strdef{x}{a}{a}\\
  &\strdef{y}{b}{x}
\end{align*}

\noindent are unsolvable for $\rhoIn$'s where $a$ and $b$ are not perfectly synchronized,
but would always have a trivial solution if pacing annotations were ignored (just set $\rhoOut(y) = \rhoOut(x) = \rhoIn(a)$).
We wish to ensure that annotated equations have solutions for \emph{all} possible inputs.
Equations that satisfy this property are called \emph{consistent}.

\begin{definition}[Consistency]
  A specification $S$ is called \emph{consistent} if, for any input map $\rhoIn$, there exists at
  least one compatible output map $\rhoOut$. Formally \[
    \mathit{Consistent}(S) \triangleq \forall \rhoIn. \ \exists \rhoOut. \ (\rhoIn, \rhoOut) \in \den{S}
  \]
\end{definition}

\subsection{Type System for Pacing Annotations}

To determine whether a given specification is consistent, we use a type system.
Our type system uses typing judgements of the following two forms \[
  \Gamma \vdash e : \tau \quad \textrm{and} \quad \Gamma \vdash S
\]

\noindent where $\Gamma$ is a \emph{typing context}, $e$ is a stream expression, $\tau$ is a pacing annotation, and $S$ is a list of annotated equations.
The first kind of judgement, $\Gamma \vdash e : \tau$, is used to check the pacing of stream expressions $e$.
Judgements $\Gamma \vdash S$ are used to check the consistency of sets of equations.
In both judgements, the typing context $\Gamma$ is a partial map from output variables $x \in \Vout$ to annotations $\tau$, representing assumptions made on the pacing of some output streams.
In the following, we use the notation convention $\Gamma, x : \tau$
to denote a context where $x$ maps to $\tau$, and we note $(x : \tau) \in \Gamma$ when $\Gamma(x) = \tau \ne \bot$.



\begin{figure}[ht!]
  \begin{mathpar}
    \textbf{Equations}\\
    \inferrule[Empty]{ }{\Gamma \vdash \epsilon} \and 
    \inferrule[Equation]{ x \notin \mathit{dom}(\Gamma)\\\Gamma \vdash e : \tau\\ \Gamma, x : \tau \vdash S }{\Gamma \vdash (\strdef{x}{\tau}{e}) \cdot S }\\
    \textbf{Expressions}\\
    \inferrule[Const]{v \in \mathbb{Z}}{\Gamma \vdash v : \tau}\and
    \inferrule[BinOp]{\Gamma \vdash e_1 : \tau_\mathit{must} \\ \Gamma \vdash e_2 : \tau_\mathit{must} }{\Gamma \vdash \stradd{e_1}{e_2} : \tau_\mathit{must}}\\
    \inferrule[DirectOut]{(x : \tau_\mathit{can}) \in \Gamma\\ \tau_\mathit{must} \models \tau_\mathit{can}}{\Gamma \vdash x : \tau_\mathit{must}}\and
    \inferrule[DirectIn]{x \in \mathbb{V}_\texttt{in} \\ \tau_\mathit{must} \models x }{\Gamma \vdash x : \tau_\mathit{must}}\\
    \inferrule[PrevOut]{(x : \tau_\mathit{can}) \in \Gamma \\ \Gamma \vdash e : \tau_\mathit{must} \\ \tau_\mathit{must} \models \tau_\mathit{can}\\ }{\Gamma \vdash \strlast{x}{e} : \tau_\mathit{must}}\and
    \inferrule[PrevIn]{x \in \Vin \\ \Gamma \vdash e : \tau_\mathit{must} \\ \tau_\mathit{must} \models x\\ }{\Gamma \vdash \strlast{x}{e} : \tau_\mathit{must}}\\
    \inferrule[HoldOut]{x \in \mathit{dom}(\Gamma)\\ \Gamma \vdash e : \tau_\mathit{must}}{\Gamma \vdash \strhold{x}{e} : \tau_\mathit{must}}\and
    \inferrule[HoldIn]{x \in \Vin \\ \Gamma \vdash e : \tau_\mathit{must} }{\Gamma \vdash \strhold{x}{e} : \tau_\mathit{must}}
  \end{mathpar}
  \caption{The \CLola typing rules for pacing types}
  \label{fig:typing-v1}
\end{figure}

\headline{Typing Equations}

The typing rules for equations are processing the stream equations in a specification $S$ in order of their appearance, extending the typing context after each processed equation.
The \textsc{Empty} rule states that an empty list of equations is consistent in any context.
The \textsc{Equation} rule processes the first equation of $\strdef{x}{\tau}{e}$ of a list.
For such an equation to be well-typed, the $x$ should be unbound in $\Gamma$ (i.e., $x \notin \mathit{dom}(\Gamma)$), making sure that there is at most one stream equation
per output variables in the specifications.
Further, we require the stream expression $e$ to be well-typed in $\Gamma$ for the annotated pacing (i.e., $\Gamma \vdash e : \tau$).
Finally, the remaining equations should be well-typed in a context where $x$ is assumed to be $\tau$-paced.

\headline{Typing Expressions}

The more interesting typing rules arise from stream accesses.
Each kind of stream access comes associated with two typing rules depending on whether the accessed variable is an input or an output.
The rules for synchronous accesses (i.e., \textsc{DirectOut}, \textsc{DirectIn}, \textsc{PrevOut}, and \textsc{PrevIn}) check that the time points at which the expression \emph{must} be evaluated according to user annotations ($\tau_\mathit{must}$) are subsumed by the time points at which the accessed stream \emph{can} be evaluated according to the current typing context ($\tau_\mathit{can}$).
This is made explicit by side conditions $\tau_\mathit{must} \models \tau_\mathit{can}$, formally defined as follows:
\[
  \tau_\mathit{must} \models \tau_\mathit{can} \triangleq \forall \rhoIn. \ \den{\tau_\mathit{must}}_\rhoIn \subseteq \den{\tau_\mathit{can}}_\rhoIn
\]
In all rules, $\tau_\mathit{must}$ corresponds to the pacing type that the target expression is expected to have.
If the accessed stream is an output stream, $\tau_\mathit{can}$ is determined from the typing context.
For accesses to input streams, $\tau_\mathit{can}$ is the stream name itself.
An important remark is that checking $\tau_\mathit{must} \models \tau_\mathit{can}$ can be done using a SAT-solver.
Indeed, since pacing annotations are positive boolean formulas, it suffices to check the validity of the boolean implication $\tau_\mathit{must} \to \tau_\mathit{can}$.

Since \syn{hold} and \syn{prev} accesses can depend on the value of a default expression $e$, their associated typing rules require the default expression to be well-paced. For example, for $x\syn{.prev}(e)$ to be $\tau_\mathit{must}$ paced,
$e$ should also be $\tau_\mathit{must}$ paced.
We note that, contrary to synchronous accesses, \syn{hold} accesses do not impose pacing restrictions on their accessed variables:
hold accesses $x\syn{.hold}(e)$ are $\tau$-paced for any $\tau$ (provided the default expression $e$ is also $\tau$-paced).
This coincides with the semantics of \syn{hold}, which is guaranteed to produce a defined value by either returning the last defined value of the accessed stream, or a default value.

\section{Examples}

We demonstrate how to use our type system by applying it to a simple example.
In the following, we suppose that inputs stream variables are $\Vin = \{ \ a, b \ \}$, and output stream variables are $\Vout = \{ x, y \}$.
In Section~\ref{sec:semantics}, we observed that the following two equations are inconsistent (because $y$ is forced to have a defined value whenever $b$ also does, even though $x$ might be undefined):
\begin{align*}
  &\strdef{x}{a}{a}\\
  &\strdef{y}{b}{x}
\end{align*}

For the purpose of the demonstration, let us try to check $\emptyset \vdash (\strdef{x}{a}{a}) \cdot (\strdef{x}{b}{x})$.
Using the rule \textsc{Equation} two times (one for each equation), we ultimately have to type-check two expressions in two different contexts: \[
  \vdash a : a \quad \textrm{and} \quad x : a \vdash x : b
\]
Since $a$ is an input variable and $a \models a$, the
first judgement is trivially handled using \textsc{DirectIn}.
However, the second judgement $x : a \vdash x : b$ cannot be handled.
Since $x$ is an output variable, the only rule we could use is \textsc{DirectOut}, but we cannot use it since $a \not\models b$. The (unsuccessful) partial derivation is given below.

\begin{mathpar}
  \inferrule*[left=\textsc{Equation}]{
    \inferrule*[left=\textsc{DirectIn}]{a \models a}{\emptyset \vdash a : a}
    \\
     \inferrule*[left=Equation]{
      \inferrule*[]{\color{red}\textit{impossible} \ (a \not\models b)}{x : a \vdash x : b}\\
      \ldots
    }{x : a \vdash \strdef{x}{b}{x}}
  }{
    \emptyset \vdash (\strdef{x}{a}{a}) \cdot (\strdef{x}{b}{x})
  }
\end{mathpar}

This inconsistency can be fixed, for example, by replacing $x$ with $\strhold{x}{b}$ in the second equation.
In this case the stream $y$ would still be guaranteed to have a value whenever a fresh input $b$ is received, but this value will only be an approximation of the stream $x$.
The full typing derivation is given below.

\begin{mathpar}
  \inferrule*[left=\textsc{Equation}]{
    \inferrule*[left=\textsc{DirectIn}]{a \models a}{\emptyset \vdash a : a}
    \\
     \inferrule*[left=Equation]{\ldots}{x : a \vdash \strdef{x}{b}{\strhold{x}{b}}}
  }{
    \emptyset \vdash (\strdef{x}{a}{a}) \cdot (\strdef{x}{b}{\strhold{x}{b}})
  }\\
  \\
  \inferrule*[left=]{
      \inferrule*[left=\textsc{HoldOut}]{
        x \in \{ \ x \ \}\\
        \inferrule*[left=\textsc{DirectIn}]{ b \models b }{x : a\vdash b : b}
      }{
        x : a \vdash \strhold{x}{b} : b
      }\\
      \inferrule*[right=\textsc{Empty}]{ }{x : a, y : b \vdash \epsilon}
    }{
      \ldots
    }
\end{mathpar}

We note that this time, we were able to use \textsc{HoldOut} to justify that the second equation $\strdef{y}{b}{\strhold{x}{b}}$ has a well-paced solution for all inputs.
We only needed to check that the default value 
in the access $\strhold{x}{b}$ (i.e., $b$) is defined at time points where $y$ must be defined.
Since the pacing annotation of the stream $y$ is $b$, this condition is trivially true ($b \models b$).
  \section{Soundness Proof}
\label{sec:proof}

Having defined our type system, we still need to formally prove
its soundness.
Namely, if a specification is well-typed according to the rules presented in the previous section, is should be \emph{consistent}.
We prove soundness by constructing an adequate \emph{semantic model} of our typing judgements.

\subsection{Semantic Model}

\newcommand{\Lsafe}[1]{\mathcal{S}\llbracket #1 \rrbracket}
\newcommand{\Lctx}[1]{\mathcal{G}\llbracket #1 \rrbracket}
\newcommand{\Lstr}[1]{\mathcal{W}\llbracket #1 \rrbracket}
\newcommand{\Lequ}[1]{\mathcal{E}\llbracket #1 \rrbracket}

\headline{Interpretation of Typing Contexts}

To construct a semantic interpretation of typing judgements $\Gamma \vdash e : \tau$ and $\Gamma \vdash S$, the first thing we need is to interpret the context $\Gamma$.
Given an input map $\rhoIn$, a typing context $\Gamma$ can be understood as describing partial stream maps whose streams are paced as described by the bindings $(x : \tau) \in \Gamma$.
We denote this interpretation by $\den{\Gamma}_\rhoIn \subseteq \mathit{Pmap}$.

\begin{definition}[Interpretation of typing contexts]
  \begin{align*}
    \den{\Gamma}_{\rhoIn} &\triangleq \{ \ \rhoOut \mid \mathit{dom}(\Gamma) = \mathit{dom}(\rhoOut) \wedge \forall (x : \tau) \in \Gamma. \ \rhoOut(x) \in \den{\tau}^\omega_{\rho_\texttt{in}} \ \}
  \end{align*}
\end{definition}

\headline{Semantic Typing of Expressions}

Intuitively, if $\Gamma \vdash e : \tau$, then for any input map $\rhoIn$, and for any output map that is compatible with $\Gamma$, $e$ should have a defined value at least at all time points that are described by $\tau$. This semantic interpretation of typing judgements for expressions is denoted by $\Gamma \vDash e : \tau$.

\newcommand{\pden}[1]{\overline{\den{#1}}}

\begin{definition}[Semantic typing of expressions]
  \begin{align*}
    \Gamma \vDash e : \tau \triangleq \forall \rhoIn. \forall \rhoOut \in \den{\Gamma}_{\rhoIn}. \forall n \in \den{\tau}_\rhoIn. \ \pden{e}_{(\rhoIn, \rhoOut)}(n) \ne \bot
  \end{align*}
\end{definition}

We note that in this definition, output maps $\rhoOut$ are only \emph{partial maps}, so we cannot use the standard denotation of expressions $\den{e}$ to evaluate $e$ in $\rhoOut$.
Instead, we use a slightly different \emph{partial semantics} $\pden{e}$, returning $\bot$ whenever $e$ contains an access to an output variable $x \notin \mathit{dom}(\rhoOut)$.
We omit the full formal definition of $\pden{e}$ here, but we note that this partial semantics \emph{agrees} with the reference semantics in the sense of the following lemma.

\begin{lemma}
  \label[plemma]{lem:pden-den}
  For any $\rhoIn \in \mathit{InMap}$, any $\rho^1_\mathit{out} \in \mathit{Pmap}$, and any $\rho^2_\mathit{out} \in \mathit{OutMap}$, \[
    \pden{e}_{(\rhoIn, \rho^1_\mathit{out})} \preccurlyeq \den{e}_{(\rhoIn, \rho^1_\mathit{out} \cdot \rho^2_\mathit{out})}
  \]
\end{lemma}

\noindent In other words, if $e$ has defined values under the partial output map $\rho^1_\mathit{out}$, then it has the same defined values under any totalization $\rho^1_\mathit{out} \cdot \rho^2_\mathit{out}$.

\headline{Semantic Typing of Equations}

As for expressions, we construct a semantic interpretation of typing judgements $\Gamma \vdash S$ for equations.

\begin{definition}[Semantic typing of equations]
  \begin{align*}
    \Gamma \vDash S \triangleq \forall \rhoIn. \forall \rho^1_\mathit{out} \in \den{\Gamma}_\rhoIn. \exists \rho^2_\mathit{out} \in \mathit{OutMap}. \ \den{S}_{(\rhoIn, \rho^1_\mathit{out} \cdot \rho^2_\mathit{out})}
  \end{align*}
\end{definition}

Intuitively, for any partial output map $\rho^1_\mathit{out}$ well-paced according to $\Gamma$, the equations in $S$ are expected to have a well-paced solution $\rho^2_\mathit{out}$.
As a direct consequence of this definition, if a set of equations $S$ is \emph{semantically} well-paced in an empty context, then it is consistent (adequacy of the semantic model).

\begin{theorem}[Adequacy]
  \label{thm:adeq}
  $\emptyset \vDash S \implies \mathit{Consistent}(S)$
\end{theorem}
\begin{proof}
  Suppose $\emptyset \vDash S$ and let $\rhoIn \in \mathit{InMap}$.
  Pick $\rho^1_\mathit{out} = (\lambda \_. \ \bot) \in \den{\emptyset}_\rhoIn$, and by definition of $\emptyset \vDash S$ there must exists a solution $\rho^2_\mathit{out}$ to the equations $S$. \qed
\end{proof}

\subsection{Soundness of Typing Rules}

Now that we have established the \emph{adequacy} of our semantic model, we can prove that our typing rules are correct with respect to semantic typing.

\begin{theorem}[Sound typing of expressions]
  \label{thm:sound-expr}
  $\Gamma \vdash e : \tau \implies \Gamma \vDash e : \tau$
\end{theorem}
\begin{proof}
  The proof proceeds by a straightforward induction on the derivation $\Gamma \vdash e : \tau$.
  We refer the curious readers to our Rocq development for details.
  \qed
\end{proof}

\begin{theorem}[Sound typing of equations]
  \label{thm:sound-equ}
  $\Gamma \vdash S \implies \Gamma \vDash S$
\end{theorem}
\begin{proof}
  The proof proceeds by induction on the derivation $\Gamma \vdash S$.
  The case $S = \epsilon$ is straightforward.
  Now suppose $S = \strdef{x}{\tau}{e} \cdot S'$,
  suppose $\Gamma \vdash e : \tau$, and by induction assume $\Gamma, x : \tau \vDash S'$.
  For a fixed $\rhoIn$ and a fixed $\rho^1_\mathit{out} \in \den{\Gamma}_\rhoIn$, we have to show the existence of a solution $\rho^2_\mathit{out}$ to $\strdef{x}{\tau}{e} \cdot S$.
  By Theorem~\ref{thm:sound-expr}, since $\Gamma \vdash e : \tau$, this implies
  $\Gamma \vDash e : \tau$.
  We can then construct a stream $w_e$ as follows: \begin{align*}
    w_e(n) \triangleq \begin{cases}
      \begin{tabular}{ll}
        $\pden{e}_{(\rhoIn, \rho^1_\mathit{out})}$ &if $n \in \den{\tau}_\rhoIn$\\
        $\bot$ &otherwise
      \end{tabular}
    \end{cases}
  \end{align*}
  By definition of $\Gamma \vDash e : \tau$, we have
  $w_e \in \den{\tau}_\rhoIn$.
  Since $\rho^1_\mathit{out} \in \den{\Gamma}$, we can deduce that $\rho^1_\mathit{out}[x \gets w_e] \in \den{\Gamma, x : \tau}$.
  Applying $\Gamma, x : \tau \vDash S'$ to $\rho^1_\mathit{out}[x \gets w_e]$, we obtain a solution $\rho'_\mathit{out}$ to $S'$.
  I.e., $(\rhoIn, \rho^1_\mathit{out}[x \gets w_e] \cdot \rho'_\mathit{out}) \in \den{S'}$.
  By definition of $w_e$ and Lemma~\ref{lem:pden-den},
  $\rho^1_\mathit{out}[x \gets w_e] \cdot \rho'_\mathit{out}$ is also a solution to the equation $\strdef{x}{\tau}{e}$.
  Overall, we have $(\rhoIn, \rho^1[x \gets w_e] \cdot \rho'_\mathit{out}) \in \den{\strdef{x}{\tau}{e} \cdot S}$, which concludes the proof.
  \qed
\end{proof}

As an immediate consequence of Theorem~\ref{thm:adeq} and Theorem~\ref{thm:sound-equ}, we obtain the soundness of our type system for consistency.

\begin{corollary}[Type Soundness]
  $\emptyset \vdash S \implies \mathit{Consistent}(S)$
\end{corollary}
  
  \section{Sound Extensions of the Type System}
\label{sec:limitations}

While the type system presented in Figure~\ref{fig:typing-v1} is
sound, it is also incomplete: it may reject equations that are consistent.
This section describes two simple extensions of the type system to mitigate incompleteness.

\headline{Ordering of Equations}

A first source of incompleteness originates from the order-sensitive nature of our type system.
As an example, the following two (consistent) equations do not type-check in this order, because $y$ is defined after $x$ and $x$ uses $y$. However, they are well-typed if we swap $x$ and $y$.

\begin{minipage}{.2\textwidth}
\begin{align*}
  &\strdef{x}{i}{y}\\
  &\strdef{y}{i}{i}
\end{align*}
\end{minipage}%
\begin{minipage}{.8\textwidth}
  \begin{align*}
  \inferrule[Reorder]{
    S' \ \textrm{is a permutation of} \ S \\ \Gamma \vdash S'
  }{\Gamma \vdash S}
\end{align*}
\end{minipage}

\medskip

\noindent This issue is easily fixed by introducing the rule \textsc{Reorder} above. This extension is shown to preserve soundness of the type system using the same semantic model as in Section~\ref{sec:proof}, and by observing that the semantics of equations is insensitive to re-orderings of the equations.

\headline{Self References}

Another source of incompleteness stems from the presence of self-references in stream equations.
Equations such as $\strdef{x}{\syn{$\top$}}{\strlast{x}{0} + 1}$ are consistent, but rejected by the typing rules of Figure~\ref{fig:typing-v1}.
However, such self \syn{prev}-access
are very useful to implement counters or other stateful computations.
To enable self \syn{prev}-accesses, we extend our type system to make it self-reference-aware. Concretely, we replace judgements $\Gamma \vdash e : \tau$ by judgements $\Gamma \vdash^x e : \tau$, where the additional annotation $x \in \Vout$ indicates the name of the stream containing $e$ as a sub-expression.
With this change, the rule for equations is updated to track self-references, and an additional rule \textsc{PrevSelf} can be added to support self \syn{prev}-accesses. \begin{mathpar}
  \inferrule[PrevSelf]{x = y \\ \Gamma \vdash^x e : \tau}{\Gamma \vdash^x \strlast{y}{e} : \tau}\qquad
  \inferrule[Equation]{ \Gamma \vdash^x e : \tau \\ \Gamma, x : \tau \vdash S }{\Gamma \vdash \strdef{x}{\tau}{e} \cdot S}
\end{mathpar}

\noindent All other rules are kept unchanged, and simply ignore the extra annotation. This extension preserves soundness, but the formal proof is significantly more complicated.
In particular, the extra annotation $x$ in the typing judgement for expressions needs a semantic interpretation. We omit the formalization here, and refer the curious readers to our Rocq development for a formal proof of soundness of this extension.

We note that beyond self \syn{prev}-accesses, other types of self-references (like self \syn{hold}-access, or equations of the form $x = x$) are considered to be anti-patterns even when they have models, and therefore do not need to be covered.

  \section{Conclusion}
In this paper, we formalized the core of \textsc{RTLola}'s pacing annotations, a mechanism for data synchronization in stream-based monitoring.
We provided a formal semantics interpreting annotated stream equations as relations between input and output streams.
We showed that pacing annotations can introduce inconsistencies and addressed this by treating pacing annotations as types.
The resulting type system rejects contradictory annotations, and we proved its soundness with respect to the formal semantics.

While the presented fragment of \textsc{RTLola} is already practical, extending it to the whole language remains future work.
Another interesting research direction would be to relate the type system to an operational semantics, establishing not only the existence of models but also their computability.

\section*{Data Availability}

Our Rocq development containing the whole formalization of \CLola, its type system (and the extensions presented in Section~\ref{sec:limitations}), and the soundness proof are available online at the following addresses: \begin{center}
  \url{https://github.com/acorrenson/pacing-types} (\textit{latest updates})\\
  or\\
  \url{https://doi.org/10.5281/zenodo.18759084} (\textit{artifact})
\end{center}

\section*{Acknowledgements}

This work was supported by the DFG in project
389792660 (TRR 248 – CPEC) and the European Research Council (ERC) Grant HYPER (No. 101055412). Views and opinions expressed are however those of the authors only and do not necessarily reflect those of the European Union or the European Research Council Executive Agency.
Neither the European Union nor the granting authority can be held responsible for them.
F.~Kohn, A.~Correnson, and J.~Baumeister carried out this work as members of the Saarbr\"ucken Graduate School of Computer Science.

  \bibliographystyle{plain}
  \bibliography{paper}

@inproceedings{DBLP:conf/fm/BaumeisterFKS24,
  author    = {Jan Baumeister and
               Bernd Finkbeiner and
               Florian Kohn and
               Frederik Scheerer},
  editor    = {Andr{\'{e}} Platzer and
               Kristin Yvonne Rozier and
               Matteo Pradella and
               Matteo Rossi},
  title     = {A Tutorial on Stream-Based Monitoring},
  booktitle = {Formal Methods - 26th International Symposium, {FM} 2024, Milan, Italy,
               September 9-13, 2024, Proceedings, Part {II}},
  series    = {Lecture Notes in Computer Science},
  volume    = {14934},
  pages     = {624--648},
  publisher = {Springer},
  year      = {2024},
  url       = {https://doi.org/10.1007/978-3-031-71177-0\_33},
  doi       = {10.1007/978-3-031-71177-0\_33},
  bibsource = {dblp computer science bibliography, https://dblp.org}
}

@inproceedings{d2005lola,
  title        = {LOLA: runtime monitoring of synchronous systems},
  author       = {d'Angelo, Ben and Sankaranarayanan, Sriram and S{\'a}nchez, C{\'e}sar and Robinson, Will and Finkbeiner, Bernd and Sipma, Henny B and Mehrotra, Sandeep and Manna, Zohar},
  booktitle    = {12th International Symposium on Temporal Representation and Reasoning (TIME'05)},
  pages        = {166--174},
  year         = {2005},
  organization = {IEEE}
}

@inproceedings{baumeister2020rtlola,
  title        = {RTLola cleared for take-off: monitoring autonomous aircraft},
  author       = {Baumeister, Jan and Finkbeiner, Bernd and Schirmer, Sebastian and Schwenger, Maximilian and Torens, Christoph},
  booktitle    = {Computer Aided Verification: 32nd International Conference, CAV 2020, Los Angeles, CA, USA, July 21--24, 2020, Proceedings, Part II 32},
  pages        = {28--39},
  year         = {2020},
  organization = {Springer}
}

@inproceedings{DBLP:conf/rv/KallwiesLSSTW22,
  author    = {Hannes Kallwies and
               Martin Leucker and
               Malte Schmitz and
               Albert Schulz and
               Daniel Thoma and
               Alexander Weiss},
  editor    = {Thao Dang and
               Volker Stolz},
  title     = {TeSSLa - An Ecosystem for Runtime Verification},
  booktitle = {Runtime Verification - 22nd International Conference, {RV} 2022, Tbilisi,
               Georgia, September 28-30, 2022, Proceedings},
  series    = {Lecture Notes in Computer Science},
  volume    = {13498},
  pages     = {314--324},
  publisher = {Springer},
  year      = {2022},
  doi       = {10.1007/978-3-031-17196-3\_20},
  bibsource = {dblp computer science bibliography, https://dblp.org}
}

@article{bauer2011runtime,
  title     = {Runtime verification for LTL and TLTL},
  author    = {Bauer, Andreas and Leucker, Martin and Schallhart, Christian},
  journal   = {ACM Transactions on Software Engineering and Methodology (TOSEM)},
  volume    = {20},
  number    = {4},
  pages     = {1--64},
  year      = {2011},
  publisher = {ACM New York, NY, USA}
}

@inproceedings{Lola2.0,
  author    = {Faymonville, Peter
               and Finkbeiner, Bernd
               and Schirmer, Sebastian
               and Torfah, Hazem},
  editor    = {Falcone, Yli{\`e}s
               and S{\'a}nchez, C{\'e}sar},
  title     = {A Stream-Based Specification Language for Network Monitoring},
  booktitle = {Runtime Verification},
  year      = {2016},
  publisher = {Springer International Publishing},
  address   = {Cham},
  pages     = {152--168},
  isbn      = {978-3-319-46982-9}
}

@inproceedings{DBLP:conf/cav/FaymonvilleFSSS19,
  author    = {Peter Faymonville and
               Bernd Finkbeiner and
               Malte Schledjewski and
               Maximilian Schwenger and
               Marvin Stenger and
               Leander Tentrup and
               Hazem Torfah},
  editor    = {Isil Dillig and
               Serdar Tasiran},
  title     = {StreamLAB: Stream-based Monitoring of Cyber-Physical Systems},
  booktitle = {Computer Aided Verification - 31st International Conference, {CAV}
               2019, New York City, NY, USA, July 15-18, 2019, Proceedings, Part
               {I}},
  series    = {Lecture Notes in Computer Science},
  volume    = {11561},
  pages     = {421--431},
  publisher = {Springer},
  year      = {2019},
  doi       = {10.1007/978-3-030-25540-4\_24},
  bibsource = {dblp computer science bibliography, https://dblp.org}
}

@inproceedings{DBLP:conf/fm/GorostiagaS21,
  author    = {Felipe Gorostiaga and
               C{\'{e}}sar S{\'{a}}nchez},
  editor    = {Marieke Huisman and
               Corina S. Pasareanu and
               Naijun Zhan},
  title     = {HStriver: {A} Very Functional Extensible Tool for the Runtime Verification
               of Real-Time Event Streams},
  booktitle = {Formal Methods - 24th International Symposium, {FM} 2021, Virtual
               Event, November 20-26, 2021, Proceedings},
  series    = {Lecture Notes in Computer Science},
  volume    = {13047},
  pages     = {563--580},
  publisher = {Springer},
  year      = {2021},
  doi       = {10.1007/978-3-030-90870-6\_30},
  bibsource = {dblp computer science bibliography, https://dblp.org}
}

@inproceedings{DBLP:conf/cav/JungesTS20,
  author    = {Sebastian Junges and
               Hazem Torfah and
               Sanjit A. Seshia},
  editor    = {Alexandra Silva and
               K. Rustan M. Leino},
  title     = {Runtime Monitors for Markov Decision Processes},
  booktitle = {Computer Aided Verification - 33rd International Conference, {CAV}
               2021, Virtual Event, July 20-23, 2021, Proceedings, Part {II}},
  series    = {Lecture Notes in Computer Science},
  volume    = {12760},
  pages     = {553--576},
  publisher = {Springer},
  year      = {2021},
  doi       = {10.1007/978-3-030-81688-9\_26},
  bibsource = {dblp computer science bibliography, https://dblp.org}
}

@inproceedings{DBLP:conf/cav/HenzingerKKM23,
  author    = {Thomas A. Henzinger and
               Mahyar Karimi and
               Konstantin Kueffner and
               Kaushik Mallik},
  editor    = {Constantin Enea and
               Akash Lal},
  title     = {Monitoring Algorithmic Fairness},
  booktitle = {Computer Aided Verification - 35th International Conference, {CAV}
               2023, Paris, France, July 17-22, 2023, Proceedings, Part {II}},
  series    = {Lecture Notes in Computer Science},
  volume    = {13965},
  pages     = {358--382},
  publisher = {Springer},
  year      = {2023},
  doi       = {10.1007/978-3-031-37703-7\_17},
  bibsource = {dblp computer science bibliography, https://dblp.org}
}

@inproceedings{DBLP:conf/cav/BasinKZ17,
  author    = {David A. Basin and
               Felix Klaedtke and
               Eugen Zalinescu},
  editor    = {Rupak Majumdar and
               Viktor Kuncak},
  title     = {Runtime Verification of Temporal Properties over Out-of-Order Data
               Streams},
  booktitle = {Computer Aided Verification - 29th International Conference, {CAV}
               2017, Heidelberg, Germany, July 24-28, 2017, Proceedings, Part {I}},
  series    = {Lecture Notes in Computer Science},
  volume    = {10426},
  pages     = {356--376},
  publisher = {Springer},
  year      = {2017},
  doi       = {10.1007/978-3-319-63387-9\_18},
  bibsource = {dblp computer science bibliography, https://dblp.org}
}

@article{DBLP:journals/pieee/HalbwachsCRP91,
  author    = {Nicolas Halbwachs and
               Paul Caspi and
               Pascal Raymond and
               Daniel Pilaud},
  title     = {The synchronous data flow programming language {LUSTRE}},
  journal   = {Proc. {IEEE}},
  volume    = {79},
  number    = {9},
  pages     = {1305--1320},
  year      = {1991},
  doi       = {10.1109/5.97300},
  bibsource = {dblp computer science bibliography, https://dblp.org}
}

@inproceedings{DBLP:conf/concur/BerryC84,
  author    = {G{\'{e}}rard Berry and
               Laurent Cosserat},
  editor    = {Stephen D. Brookes and
               A. W. Roscoe and
               Glynn Winskel},
  title     = {The {ESTEREL} Synchronous Programming Language and its Mathematical
               Semantics},
  booktitle = {Seminar on Concurrency, Carnegie-Mellon University, Pittsburg, PA,
               USA, July 9-11, 1984},
  series    = {Lecture Notes in Computer Science},
  volume    = {197},
  pages     = {389--448},
  publisher = {Springer},
  year      = {1984},
  doi       = {10.1007/3-540-15670-4\_19},
  bibsource = {dblp computer science bibliography, https://dblp.org}
}

@inproceedings{volostream,
  author    = {Jan Baumeister and
               Bernd Finkbeiner and
               Florian Kohn and
               Florian L{\"{o}}hr and
               Guido Manfredi and
               Sebastian Schirmer and
               Christoph Torens},
  editor    = {Arie Gurfinkel and
               Vijay Ganesh},
  title     = {Monitoring Unmanned Aircraft: Specification, Integration, and Lessons-Learned},
  booktitle = {Computer Aided Verification - 36th International Conference, {CAV}
               2024, Montreal, QC, Canada, July 24-27, 2024, Proceedings, Part {II}},
  series    = {Lecture Notes in Computer Science},
  volume    = {14682},
  pages     = {207--218},
  publisher = {Springer},
  year      = {2024},
  doi       = {10.1007/978-3-031-65630-9\_10},
  bibsource = {dblp computer science bibliography, https://dblp.org}
}

@article{DBLP:journals/pacmpl/ChenMABJSWZJ24,
  author    = {Jiawei Chen and
               Jos{\'{e}} Luiz Vargas de Mendon{\c{c}}a and
               Bereket Ayele and
               Bereket Ngussie Bekele and
               Shayan Jalili and
               Pranjal Sharma and
               Nicholas Wohlfeil and
               Yicheng Zhang and
               Jean{-}Baptiste Jeannin},
  title     = {Synchronous Programming with Refinement Types},
  journal   = {Proc. {ACM} Program. Lang.},
  volume    = {8},
  number    = {{ICFP}},
  pages     = {938--972},
  year      = {2024},
  url       = {https://doi.org/10.1145/3674657},
  doi       = {10.1145/3674657},
  bibsource = {dblp computer science bibliography, https://dblp.org}
}

@inproceedings{DBLP:conf/hybrid/BenvenisteBCPP14,
  author    = {Albert Benveniste and
               Timothy Bourke and
               Beno{\^{\i}}t Caillaud and
               Bruno Pagano and
               Marc Pouzet},
  editor    = {Martin Fr{\"{a}}nzle and
               John Lygeros},
  title     = {A type-based analysis of causality loops in hybrid systems modelers},
  booktitle = {17th International Conference on Hybrid Systems: Computation and Control
               (part of {CPS} Week), HSCC'14, Berlin, Germany, April 15-17, 2014},
  pages     = {71--82},
  publisher = {{ACM}},
  year      = {2014},
  url       = {https://doi.org/10.1145/2562059.2562125},
  doi       = {10.1145/2562059.2562125},
  bibsource = {dblp computer science bibliography, https://dblp.org}
}

@inproceedings{DBLP:conf/emsoft/BenvenisteBCP11,
  author    = {Albert Benveniste and
               Timothy Bourke and
               Beno{\^{\i}}t Caillaud and
               Marc Pouzet},
  editor    = {Samarjit Chakraborty and
               Ahmed Jerraya and
               Sanjoy K. Baruah and
               Sebastian Fischmeister},
  title     = {A hybrid synchronous language with hierarchical automata: static typing
               and translation to synchronous code},
  booktitle = {Proceedings of the 11th International Conference on Embedded Software,
               {EMSOFT} 2011, part of the Seventh Embedded Systems Week, ESWeek 2011,
               Taipei, Taiwan, October 9-14, 2011},
  pages     = {137--148},
  publisher = {{ACM}},
  year      = {2011},
  url       = {https://doi.org/10.1145/2038642.2038664},
  doi       = {10.1145/2038642.2038664},
  bibsource = {dblp computer science bibliography, https://dblp.org}
}

@inproceedings{DBLP:conf/focs/Pnueli77,
  author    = {Amir Pnueli},
  title     = {The Temporal Logic of Programs},
  booktitle = {18th Annual Symposium on Foundations of Computer Science, Providence,
               Rhode Island, USA, 31 October - 1 November 1977},
  pages     = {46--57},
  publisher = {{IEEE} Computer Society},
  year      = {1977},
  url       = {https://doi.org/10.1109/SFCS.1977.32},
  doi       = {10.1109/SFCS.1977.32},
  bibsource = {dblp computer science bibliography, https://dblp.org}
}

@inproceedings{DBLP:conf/ecrts/KimVBKLS99,
  author    = {Moonjoo Kim and
               Mahesh Viswanathan and
               Han{\^{e}}ne Ben{-}Abdallah and
               Sampath Kannan and
               Insup Lee and
               Oleg Sokolsky},
  title     = {Formally specified monitoring of temporal properties},
  booktitle = {11th Euromicro Conference on Real-Time Systems {(ECRTS} 1999), 9-11
               June 1999, York, England, UK, Proceedings},
  pages     = {114--122},
  publisher = {{IEEE} Computer Society},
  year      = {1999},
  url       = {https://doi.org/10.1109/EMRTS.1999.777457},
  doi       = {10.1109/EMRTS.1999.777457},
  bibsource = {dblp computer science bibliography, https://dblp.org}
}

@inproceedings{DBLP:conf/icse/JinMLR12,
  author    = {Dongyun Jin and
               Patrick O'Neil Meredith and
               Choonghwan Lee and
               Grigore Rosu},
  editor    = {Martin Glinz and
               Gail C. Murphy and
               Mauro Pezz{\`{e}}},
  title     = {JavaMOP: Efficient parametric runtime monitoring framework},
  booktitle = {34th International Conference on Software Engineering, {ICSE} 2012,
               June 2-9, 2012, Zurich, Switzerland},
  pages     = {1427--1430},
  publisher = {{IEEE} Computer Society},
  year      = {2012},
  url       = {https://doi.org/10.1109/ICSE.2012.6227231},
  doi       = {10.1109/ICSE.2012.6227231},
  bibsource = {dblp computer science bibliography, https://dblp.org}
}

@inproceedings{DBLP:conf/fpca/GautierG87,
  author    = {Thierry Gautier and
               Paul {Le Guernic}},
  editor    = {Gilles Kahn},
  title     = {{SIGNAL:} {A} declarative language for synchronous programming of
               real-time systems},
  booktitle = {Functional Programming Languages and Computer Architecture, Portland,
               Oregon, USA, September 14-16, 1987, Proceedings},
  series    = {Lecture Notes in Computer Science},
  volume    = {274},
  pages     = {257--277},
  publisher = {Springer},
  year      = {1987},
  url       = {https://doi.org/10.1007/3-540-18317-5\_15},
  doi       = {10.1007/3-540-18317-5\_15},
  bibsource = {dblp computer science bibliography, https://dblp.org}
}

@inproceedings{DBLP:conf/emsoft/ColacoP03,
  author    = {Jean{-}Louis Cola{\c{c}}o and
               Marc Pouzet},
  editor    = {Rajeev Alur and
               Insup Lee},
  title     = {Clocks as First Class Abstract Types},
  booktitle = {Embedded Software, Third International Conference, {EMSOFT} 2003,
               Philadelphia, PA, USA, October 13-15, 2003, Proceedings},
  series    = {Lecture Notes in Computer Science},
  volume    = {2855},
  pages     = {134--155},
  publisher = {Springer},
  year      = {2003},
  url       = {https://doi.org/10.1007/978-3-540-45212-6\_10},
  doi       = {10.1007/978-3-540-45212-6\_10},
  bibsource = {dblp computer science bibliography, https://dblp.org}
}

@inproceedings{DBLP:conf/mpc/MandelPP10,
  author    = {Louis Mandel and
               Florence Plateau and
               Marc Pouzet},
  editor    = {Claude Bolduc and
               Jos{\'{e}}e Desharnais and
               B{\'{e}}chir Ktari},
  title     = {Lucy-n: a n-Synchronous Extension of Lustre},
  booktitle = {Mathematics of Program Construction, 10th International Conference,
               {MPC} 2010, Qu{\'{e}}bec City, Canada, June 21-23, 2010. Proceedings},
  series    = {Lecture Notes in Computer Science},
  volume    = {6120},
  pages     = {288--309},
  publisher = {Springer},
  year      = {2010},
  url       = {https://doi.org/10.1007/978-3-642-13321-3\_17},
  doi       = {10.1007/978-3-642-13321-3\_17},
  bibsource = {dblp computer science bibliography, https://dblp.org}
}

@inproceedings{DBLP:conf/fm/PerezGD24,
  author    = {Ivan Perez and
               Alwyn E. Goodloe and
               Frank Dedden},
  editor    = {Andr{\'{e}} Platzer and
               Kristin Yvonne Rozier and
               Matteo Pradella and
               Matteo Rossi},
  title     = {Runtime Verification in Real-Time with the Copilot Language: {A} Tutorial},
  booktitle = {Formal Methods - 26th International Symposium, {FM} 2024, Milan, Italy,
               September 9-13, 2024, Proceedings, Part {II}},
  series    = {Lecture Notes in Computer Science},
  volume    = {14934},
  pages     = {469--491},
  publisher = {Springer},
  year      = {2024},
  url       = {https://doi.org/10.1007/978-3-031-71177-0\_27},
  doi       = {10.1007/978-3-031-71177-0\_27},
  bibsource = {dblp computer science bibliography, https://dblp.org}
}

\end{document}